\documentclass[11pt,onecolumn,english]{IEEEtran} 

\usepackage{amsfonts}
\usepackage{amssymb, amsthm}
\usepackage{graphicx}
\usepackage{hyperref}
\hypersetup{colorlinks,
   colorlinks,
    citecolor=red,
    filecolor=green,
  linkcolor=blue,
    linktoc=page,
   urlcolor=blue
}
\usepackage[square,numbers, sort&compress]{natbib}
\usepackage[dvipsnames]{xcolor}
\usepackage{amsmath,  amssymb}

\def\eps{{\varepsilon}}

\newtheorem{example}{Example}%[section]
\usepackage[margin=1in]{geometry}
\newenvironment{enumerate*}%
  {\begin{enumerate}%
    \setlength{\itemsep}{0pt}%
    \setlength{\parskip}{0pt}
    \vspace*{-\topsep}
  }%
  {\end{enumerate}
\vspace*{-1.5ex}}

\newenvironment{itemize*}%
{\vspace*{0.1ex}
  \begin{itemize}%
    \setlength{\itemsep}{0pt}%
    \setlength{\parskip}{0pt}
    \vspace*{-\topsep}
  }%
  {\end{itemize}
  \vspace*{-1.5ex}
}

\newtheorem{theorem}{Theorem}
\newtheorem{definition}{Definition}
\newtheorem{lemma}{Lemma}

\newtheorem{proposition}{Proposition}

\theoremstyle{definition}

\newtheorem{remark}{Remark}

\title{Differentially Private Mechanisms for Count Queries$^*$\footnote{$^*$This is a working paper. Comments to improve this work are welcome.}} 
\author{Parastoo Sadeghi$^\dag$, Shahab Asoodeh$^\ddag$ and Flavio du Pin Calmon$^\ddag$\\$^\dag$RSEEME, Australian National University, Australia \\$^\ddag$ School of Engineering and Applied Sciences, Harvard University, USA \\ Emails: parastoo.sadeghi@anu.edu.au, \{shahab, flavio\}@seas.harvard.edu}

\begin{document}
\maketitle
\begin{abstract}
In this paper, we consider the problem of responding to a count query (or any other integer-valued queries) evaluated on a dataset containing sensitive attributes. 
To protect the privacy of individuals in the dataset, a standard practice is to add continuous noise to the true count. We design a differentially-private mechanism which adds integer-valued noise allowing the released output to remain integer. As a trade-off between utility and privacy, we derive privacy parameters $\eps$ and $\delta$ in terms of the the probability of releasing an erroneous count under the assumption that the true count is no smaller than half the support size of the noise. We then numerically demonstrate that our mechanism provides higher privacy guarantee compared to the discrete Gaussian mechanism that is recently proposed in the literature. 
\end{abstract}

\section{Introduction}
Conducting nationwide censuses is one of the most important tasks of national or federal agencies around the world. For example, the 2020 Census in the US is currently under way, which is conducted by the US Bureau of the Census every ten years \cite{us:census}.  In Australia, census is conducted by the Australian Bureau of Statistics (ABS) every five years and the next round is due in 2021 \cite{abs:census}. Responding to count queries about census outcomes is one of the most important tasks of such agencies. This can be important for distributing national funds, determining optimal voting districts and conducting various scientific or economic studies of the population. However, this might lead to dire privacy compromises.   

In order to protect the privacy of individuals in the census records, the true census results are heavily guarded. Instead, noisy versions may be released to the public in a one-off general-purpose publication or privately disclosed to a data analyst in response to specific queries. There is clearly a tension between releasing an accurate and reliable version of the census outcome and protecting the privacy of individuals in a population. This tension gets aggravated particularly for underrespresented groups in small towns and counties.  Oftentimes, their livelihoods and future funding depends on them being accurately accounted for in published census results. At the same, due to being very small in size, they are at the greatest risk of being re-identified even if perturbed versions of their data is published. For an interesting coverage of this issue in the 2010 and 2020 US Census rounds, the reader is referred to a recent article published in the New York Times \cite{nyt:article}.

Motivated by these applications, in this paper we are interested in formulating and studying optimal tradeoffs between utility and privacy when a data analyst sends a query to a data curator about the number of individuals with a common sensitive attribute in a large dataset. Assuming $n$ is the true count, the data curator releases a random variable $Y = n+Z$, where $Z$ is an integer-valued noise variable with a suitably-chosen probability distribution. 

It is not surprising that the celebrated differential privacy (DP) framework \cite{dwork} has been applied to integer-valued mechanisms for count queries. In fact, DP is a strong candidate for publishing private 2020 US Census outcomes \cite{us:census:dp}. Other countries such as Australia and New Zealand are closely watching this space and have already taken steps in evaluating the performance of their existing mechanisms ``from the lens of differential privacy" \cite{abs:lens}. 

A standard approach to preserving $\eps$-DP is to perturb the count query by adding
random noise $Z$ with Laplacian distribution whose variance is proportional to a certain property of the query (i.e., sensitivity). Adding continuous noise to an integer-valued query makes the output less interpretable.  Geometric distribution, as the discrete counterpart of Laplace distribution, was shown in \cite{gm:paper} to be "optimal" for integer-valued queries with unit sensitivity. The optimality was defined in terms of minimizing the expected cost function (chosen from a family of utility functions with some mild regularity conditions) under a Bayesian framework. Their result, in particular, implies that adding (truncated or folded) geometric noise to the above-mentioned counting query preserves DP while minimizing the probability of error, $\Pr(Z \neq 0)$. This setting was extended to the mini-max setting in \cite{Geometric_Minimax} and to queries with general sensitivity and under the worst case setting in \cite{Viswanath_DP_Staircase}. 
The design and implementation of geometric distribution on finite-precision machines were investigated in  \cite{dp:finite:computers:arxiv}.  

More recently,  the problem of minimizing the $L_p$-norm of error of the query output, subject to $\eps$-DP, was revisited in \cite{count:dp;arxiv}. It was argued that despite minimizing the probability of error, the geometric mechanism suffers from a number of limitations. Most importantly, the error $p(Y \neq n)$ depends on the true count $n$, leading to uneven accuracy for different queries. In the very high-privacy regime when $\eps \to 0$, this unevenness in the distribution results in overwhelmingly reporting two extreme values $Y = 0$ or $Y = \max~n$ (zero count or maximum possible count in the database). To address this issue, they proposed an explicit fair mechanism (EM).
However, the noise in the EM method adds bias, that is, $\mathbb E[Y|n]\neq n$. Moreover, EM (and also truncated geometric) mechanism does not restrict the support size of the noise. As such, the output can be any number between the two extremes $y = 0$ and $y = \max n$ with non-zero probability. 

In a different, and quite recent, line of work, a \textit{discrete} version of the Gaussian density was considered for designing privacy-preserving mechanism for count queries. In particular, it was shown in  \cite{count:dp:gaussian:arxiv} that adding ``discrete" Gaussian noise provides similar privacy and utility guarantees to those  obtained by the well-understood Gaussian mechanism.  Discrete Gaussian noise is supported over $\mathbb Z$ and thus the output of such mechanism can be all integers regardless of the domain of the query. 
This might lead to unpleasant outcomes in some practical scenarios. For example, according to the New York Times article \cite{nyt:article}, the population of a county in the 2010 US Census was over-reported by a factor of almost 8 (the true count was around 90, but was reported to be around 920). While post-processing (truncating or folding the probability mass function) can be used to limit the range of discrete Gaussian mechanisms, the question remains whether one can explicitly incorporate a finite support of the noise  into the mechanism design process. 

In this paper, we adopt a different approach than \cite{count:dp;arxiv, count:dp:gaussian:arxiv}. We begin by enforcing three basic desired constraints on the noise distribution. These properties include zero bias, fixed hard support, and fixed error probability irrespective of the true count $n$. By varying the probability of error we can trade off accuracy with privacy. The hard support for $Z$ ensures that for any given $n$, the output $n + Z$ takes value in a pre-specified range with probability one; in particular $Y$ is required to be always non-negative. 
We first design a noise distribution satisfying above constraints and then use it to construct a differentially-private mechanism (under a mild condition on the dataset). The constraint on the noise support has two implications. First, the values of the noise added must depend on the true count $n$ (to ensure the non-negativity of the output). Second, the mechanism provides a slightly relaxed privacy guarantee: it is $\eps$-DP with probability $1-\delta$ for some parameter $\delta>0$. This framework is known as \textit{approximate} DP and usually denoted by $(\eps, \delta)$-DP \cite{algorithmic, Dwork_Calibration}.  While the proposed noise has similar distribution as discrete Gaussian (with slightly heavier tail),  our numerical findings indicate that the privacy performance of our mechanism (in terms of $\eps$ and $\delta$) is tighter than what would be obtained via discrete Gaussian mechanism with the same variance.

We can summarize our observations about the interplay of privacy and utility parameters as follows. As expected, a lower privacy requirement, i.e., higher $\eps$, results in lower $\delta$ when the noise support size and probability of releasing an erroneous value are fixed. By moderately increasing the noise support size while keeping $\eps$ and error probability fixed, we can also significantly improve $\delta$. For example, increasing the noise support from $2D+1 = 9$ to $2D'+1 = 13$  we obtain significantly smaller $\delta$. It turns out that when $\eps$ and the noise support size are fixed, reducing the error probability (and hence improving data reliability) can be tolerated up a certain threshold without significantly affecting $\delta$. However, further reducing the error probability below such a threshold can have severe effects on the privacy performance through increasing $\delta$ beyond typically acceptable levels.

Notation: For an integer $N \in \mathbb N$, we use $[N]$ to denote the set $\{1, 2, \cdots, N\}$. For $a, b \in \mathbb Z$ and $a \leq b$, we use $[a:b]$ to denote the set $\{a, a+1, \cdots, b\}$.

\section{Problem Formulation}
Suppose a dataset of size at most $N$ is given.
A query on the number of entries in the database that possess a certain sensitive attribute is received by the data curator.  Assume $n\in [N]$ is the actual (positive) count of such attribute. To protect the privacy of individuals in the dataset, the data curator releases a noisy version of $n$ as
\begin{align}\label{eq:model}
Y = n + Z,    
\end{align}
where $Z$ is an integer-valued noise variable. We assume the noise distribution $p_Z(\cdot|n)$ satisfies the following intuitive properties:

\begin{enumerate}
\item[\textbf{P1.}] Fix $D \in \mathbb N$. Noise variable $Z$ takes values in $[-A:D]$, where $A \doteq \min\{n,D\}$, i.e., $p_Z(z|n) = 0$ if $z \notin [-A:D]$. This properties ensures that the mechanism's output is always non-negative. A smaller $D$ implies a smaller noise support set, which in turn implies a larger accuracy in reporting the query.\footnote{Note that the support of the output for a given $n$ is limited to $[n-A: n+D]$. For $\max\{1,N-D+1\}\leq n \leq N$, we allow the output to be greater than the maximum possible real count $N$ (by at most $D$). If $N \gg D$, this overshoot is negligible.}
\item[\textbf{P2.}] Fix $\eta\in (0,1)$. The probability of releasing the correct value is set to
    $$p_Y(Y = n) = \eta,$$
irrespective of the true count $n$. Equivalently, the probability of releasing an erroneous value is $$p_Y(Y \neq n) = 1-\eta \doteq \bar\eta.$$ A larger $\eta$ means more reliability in releasing the correct value. Conversely, a smaller $\eta$ means better privacy in the sense of concealing the true count of the sensitive attribute in the dataset.
\item[\textbf{P3.}] In addition, the error must have zero bias. That is, $\mathbb E[Z|n] = 0$, which is a statistically desirable property.\footnote{Note that if $n = 0$, maintaining zero bias is not possible jointly with $\eta < 1$ and P1. Hence, we limit ourselves to queries with $n \geq 1$.}
\end{enumerate}
\bigskip
The following proposition provides a simple family of noise distribution that satisfy these properties. 
\begin{proposition}\label{prop:noise}
Let the true count for a query be $n\in [N]$. For any $i_1 \in [-A:-1]$ and $i_2\in [D]$,
the following distribution satisfies P1-P3
\begin{align}\label{eq:noise:gen}
p^{(i_1, i_2)}_Z(z|n) = \begin{cases} \bar\eta\frac{i_2}{i_2+|i_1|}, \quad & z = i_1,\\
\eta, \quad & z = 0,\\
\bar\eta\frac{|i_1|}{i_2+|i_1|}, \quad & z = i_2,\\
0, \quad & \text{otherwise}.
\end{cases}
\end{align}
\end{proposition}
In light of this proposition, we can obtain that any convex combination of $p^{(i_1, i_2)}_Z(z|n)$, i.e.,  
\begin{align}\label{eq:noise:convex}
p_Z(z|n)\doteq \sum_{i_1 = -A}^{-1}\sum_{i_2= 1}^D\alpha_{i_1,i_2,n}\,\,p^{(i_1, i_2)}_Z(z|n),
\end{align}
with non-negative coefficients $ \alpha_{i_1,i_2,n}$ satisfying $\sum_{i_1 = -A}^{-1}\sum_{i_2= 1}^D\alpha_{i_1,i_2,n} = 1$, also meets P1-P3. We can thus use this convex combination to obtain a general \textit{data-dependent} noise distribution satisfying P1-P3. The following example makes this construction clear.  

\begin{example}\label{example1}
Let $N = 3$, $\eta = \frac{1}{2}$ and $D = 2$. Note that for $n = 1$, two distributions are possible for $p^{(i_1, i_2)}_Z(z|1)$ by choosing $(i_1, i_2) = (-1, 1)$ and $(i_1, i_2) = (-1, 2)$, as described below
\begin{equation}
p^{(-1, 1)}_Z(z|1) = \begin{cases} \frac{1}{4}, \quad & z = -1, 1\\
\frac{1}{2}, \quad & z = 0,\\
0, \quad & \text{otherwise},
\end{cases}
\qquad \text{and}\qquad   p^{(-1, 2)}_Z(z|1) = \begin{cases} \frac{1}{3}, \quad & z = -1,\\
\frac{1}{2}, \quad & z = 0,\\
\frac{1}{6}, \quad & z = 2,\\
0, \quad & \text{otherwise}.
\end{cases}
\end{equation}
For $n = 2$ and $n=3$, four distributions are similarly possible by combining $i_1 = -2, -1$ with $i_2 = 1, 2$. Overall, the general convex structure in \eqref{eq:noise:convex} produces the noise distributions $p_Z(\cdot|1), p_Z(\cdot|2),$ and $p_Z(\cdot|3)$ given by 
\begin{align}
p_Z(z|n) = \left(\begin{matrix}
n=1&&n=2&&n=3\\\hline\\
0&&\frac{1}{4}\alpha_{-2,2,2}+\frac{1}{6}\alpha_{-2,1,2}&&\frac{1}{4}\alpha_{-2,2,3}+\frac{1}{6}\alpha_{-2,1,3}\\
\frac{1}{4}\alpha_{-1,1,1}+\frac{1}{3}\alpha_{-1,2,1}&&\frac{1}{4}\alpha_{-1,1,2}+\frac{1}{3}\alpha_{-1,2,2}&&\frac{1}{4}\alpha_{-1,1,3}+\frac{1}{3}\alpha_{-1,2,3}\\
\frac{1}{2}&&\frac{1}{2}&&\frac{1}{2}\\
\frac{1}{4}\alpha_{-1,1,1}&&\frac{1}{4}\alpha_{-1,1,2}+\frac{1}{3}\alpha_{-2,1,2}&&\frac{1}{4}\alpha_{-1,1,3}+\frac{1}{3}\alpha_{-2,1,3}\\
\frac{1}{6}\alpha_{-1,2,1}&&\frac{1}{4}\alpha_{-2,2,2}+\frac{1}{6}\alpha_{-1,2,2}&&\frac{1}{4}\alpha_{-2,2,3}+\frac{1}{6}\alpha_{-1,2,3}
\end{matrix}\right).
\end{align}
 Each column of this matrix is a valid noise distribution for a given $n \in [N]$, subject to $\sum_{i_1 = -A}^{-1}\sum_{i_2= 1}^D\alpha_{i_1,i_2,n} = 1$. Note that row indices range from $z=-D$ to $z=D$. Also, note that $p_Z(-2|1) = 0$. This is because for $n=1$, we have $A = \min\{n,D\} = 1$, dictating $[-1:2]$ as the support set of $p_1$, instead of $[-2:2]$. Moreover,  $p_Z(z|n) = 0 $ for $n\geq 4$ or $z\notin [-2:2]$. By adding such noise to the true count $n$, the distribution of $Y$, given $n$, becomes
\begin{align}
p_Y(y|n) = \left(\begin{matrix}
n=1&&n=2&&n=3\\\hline\\
\frac{1}{4}\alpha_{-1,1,1}+\frac{1}{3}\alpha_{-1,2,1}&&\frac{1}{4}\alpha_{-2,2,2}+\frac{1}{6}\alpha_{-2,1,2}&&0\\
\frac{1}{2}&&\frac{1}{4}\alpha_{-1,1,2}+\frac{1}{3}\alpha_{-1,2,2}&&\frac{1}{4}\alpha_{-2,2,3}+\frac{1}{6}\alpha_{-2,1,3}\\
\frac{1}{4}\alpha_{-1,1,1}&&\frac{1}{2}&&\frac{1}{4}\alpha_{-1,1,3}+\frac{1}{3}\alpha_{-1,2,3}\\
\frac{1}{6}\alpha_{-1,2,1}&&\frac{1}{4}\alpha_{-1,1,2}+\frac{1}{3}\alpha_{-2,1,2}&&\frac{1}{2}\\
0&&\frac{1}{4}\alpha_{-2,2,2}+\frac{1}{6}\alpha_{-1,2,2}&&\frac{1}{4}\alpha_{-1,1,3}+\frac{1}{3}\alpha_{-2,1,3}\\
0&&0&&\frac{1}{4}\alpha_{-2,2,3}+\frac{1}{6}\alpha_{-1,2,3}
\end{matrix}\right),
\end{align}
where the rows range from $y = 0$ to $y = N+D= 5$. The value of $p_Y(y|n)$ for all other values of $y$ and $n$ not shown above is identical to zero. 
\end{example}
\subsection{Incorporating $(\eps, \delta)$-DP Conditions}
Consider the mechanism $Y = n+ Z$ with $Z$ being a data-dependant noise with distribution $p_Z(\cdot|n)$ given in \eqref{eq:noise:convex}. Our objective is to determine the smallest $\eps\geq 0$ and $\delta\in (0,1)$ such that this mechanism is  $(\eps, \delta)$-differentially private (DP) \cite{algorithmic}. 
\begin{definition}\label{Def:DP}
 Given $\eps\geq 0$ and $\delta\in [0,1)$,  a randomized algorithm $M: \mathcal X^N \to \mathcal Y$ is said to be  $(\eps, \delta)$-differentially private (DP) if  $p(M(x)\in S) \leq e^\eps p(M(x')\in S)+\delta$ for all $S \subset \mathcal Y$ and all neighboring datasets $x,x' \in \mathcal X^N$ differing in one element.
\end{definition}

Below we specialize this standard definition for our count query setup. Given an integer-valued query $q:\mathcal X^N\to [N]$ about a dataset $x$, mechanism $q$ randomizes the true query response, say $n$, via a data-dependent additive noise process, i.e., $q(x) = n + Z$. Since the query is integer-valued,  we have $|q(x)-q(x')|\leq 1$ for all pairs of neighboring datasets $x$ and $x'$. Without loss of generality, in the following we assume that $|q(x)-q(x')| = 1$ for neighboring $x$ and $x$ in the computation of DP parameters. Assuming $q(x) = n$, then it implies that we need to verify the DP requirement for dataset $x'$ for which $q(x') = n+1$ or $q(x') = n-1$. 
\begin{remark}\label{remark:DP_Singular}
Our approach to computing the DP parameters of mechanism \eqref{eq:pyn} is as follows. We first  consider Definition~\ref{Def:DP} for the \emph{singular} events $S \subset [0:N+D]$, i.e.,  all subsets $S$ satisfying $|S| =1$ and determine the corresponding parameters $\eps$ and $\delta$. Clearly, these parameters differ from the DP parameters. We then convert these parameters to a valid set of DP parameters. In the following, we use  $\eps$ and $\delta$ to denote the parameters of mechanism \eqref{eq:pyn} when applying Definition~\ref{Def:DP} for singular events $S$. 
\end{remark}

To formulate our goal, we fix $\eps \geq 0 $ and consider the following linear program
\begin{align}
\delta^*\doteq &~\min~ \delta,\nonumber\\
&~~\text{s.t.}~
p_Y(y|n) \leq e^\eps p_Y(y|n+1) + \delta, ~~\forall n\in [1:N-1], ~~\forall y\in [0: N+D],\\
&~~~~~~p_Y(y|n) \leq e^\eps p_Y(y|n-1) + \delta, ~~ \forall  n\in [2:N], ~~\forall  y\in [0:N+D].
\end{align}
Given the mechanism  $Y = n+Z$ and Proposition~\ref{prop:noise}, we can explicitly write the optimization of $\delta$ in terms of $\alpha_{i_1,i_2,n}$ as follows.
\begin{proposition}
For given $\eps \geq 0$, the optimal value $\delta^*$ is the solution of the following linear program 
\begin{align}
&\min \delta\label{Eq:LP_DP}\\
&~~\text{s.t.}\nonumber\\
\sum_{i_1, i_2}\alpha_{i_1,i_2,n}\,\,p^{(i_1, i_2)}_{Z}(y-n|n)&\leq e^\eps \sum_{i_1, i_2}\alpha_{i_1,i_2,n+1}\,\,p^{(i_1, i_2)}_{Z}(y-n-1|n+1) + \delta, ~\forall n\in [1: N-1],\\
\sum_{i_1, i_2}\alpha_{i_1,i_2,n}\,\,p^{(i_1, i_2)}_{Z}(y-n|n) &\leq e^\eps \sum_{i_1, i_2}\alpha_{i_1,i_2,n-1}\,\,p^{(i_1, i_2)}_{Z}(y-n+1|n-1) + \delta, ~\forall n\in [2: N],\\
\alpha_{i_1,i_2,n} &\geq 0, \qquad \forall i_1 \in [-A:-1],~\forall i_2 \in [D], ~\forall n \in [1:N],\\
\sum_{i_1 = -A}^{-1}\sum_{i_2= 1}^D\alpha_{i_1,i_2,n} &= 1, \qquad \forall n \in [1:N],
\end{align}
where $p^{(i_1, i_2)}_{Z}(\cdot|n)$ was defined in \eqref{eq:noise:gen} and $ y \in [0: N+D]$.
\end{proposition}
Clearly, this above optimization problem is a linear program. If the computational complexity allows, an explicit expression for  $\delta^*$ can be obtained through Fourier-Motzkin elimination of all $\alpha_{i_1,i_2,n}$. Nevertheless, we show in the next section that such expression can be derived directly if we make a minor assumption on $D$.

\section{Explicit Solutions for $n \geq D$}
 In this section, we show that if the true count satisfies $n \geq D$, then $\delta^*$  can be readily derived.
 
\begin{proposition}\label{prop:noiseD} If noise support parameter $D$ satisfies $D\leq n$, then the following data-independent noise distribution satisfies P1-P3
\begin{align}\label{noise:ngeqD}
p_Z(z) = \begin{cases} \alpha_i\frac{\bar\eta}{2}, \quad & z = -i,i\\
\eta, \quad & z = 0,\\
0, \quad & \text{otherwise},
\end{cases}
\end{align}
for all $i\in [1: D]$ and $\alpha_i \geq 0$ such that $\sum_{i=1}^D\alpha_{i} = 1$. In this case, 
\begin{align}
p_Y(y|n) = \begin{cases} \alpha_i\frac{\bar\eta}{2}, \quad & y = n-i,n+i\\
\eta, \quad & y = n,\\
0, \quad & \text{otherwise}.
\end{cases}
\end{align}
\end{proposition}
Notice that the assumption $n\geq D$ enables us to design a noise distribution that satisfies P1-P3 while being  independent of the query (i.e., $n$). 
The simple mechanism given in Proposition~\ref{prop:noiseD} can be represented by the matrix

\begin{align}\label{eq:pyn}
p_Y(y|n) = \left(\begin{matrix}
n=D&&n=D+1&&\cdots&&n = N\\\hline\\
\alpha_D \frac{\bar\eta}{2}&&0&&\cdots&&0\\
\alpha_{D-1} \frac{\bar\eta}{2}&&\alpha_D \frac{\bar\eta}{2}&&\cdots&&0\\
\cdots&&\cdots&&\cdots&&\cdots\\
\alpha_{1}\frac{\bar\eta}{2}&&\alpha_{2}\frac{\bar\eta}{2}&&\cdots&&0\\
\eta&&\alpha_{1}\frac{\bar\eta}{2}&&\cdots&&\alpha_D \frac{\bar\eta}{2}\\
\alpha_{1}\frac{\bar\eta}{2}&&\eta&&\cdots&&\alpha_{D-1} \frac{\bar\eta}{2}\\
\cdots&&\cdots&&\cdots&&\cdots\\
\alpha_D\frac{\bar\eta}{2}&&\alpha_{D-1} \frac{\bar\eta}{2}&&\cdots&&\eta\\
0&&\alpha_D\frac{\bar\eta}{2}&&\cdots&&\alpha_{1}\frac{\bar\eta}{2}\\
0&&0&&\cdots&&\cdots\\
0&&0&&0&&\alpha_D\frac{\bar\eta}{2}
\end{matrix}\right),
 \end{align}
where for exposition we have assumed $N=2D$. Each column of this matrix represents $p_{Y}(y|n)$ for $y\in [0: 3D]$.

\subsection{Lower bounds on $\delta^*$}
To compute $\delta^*$ for mechanism \eqref{eq:pyn}, we first derive a lower bound for  $\delta^*$ (this subsection) and then show that it is in fact achievable (next subsection). For notational brevity, define $E \doteq e^\eps$. The main tool for deriving the lower bound is the following proposition.    
\begin{proposition}\label{prop:dp}
For the noise distribution given by Proposition~\ref{prop:noiseD}, $\delta^*$ satisfies the following inequalities
\begin{align}
\frac{\bar\eta}{2} \alpha_D &\leq \delta^*,\label{eq:D:1}\\
\frac{\bar\eta}{2} \alpha_{i} &\leq E \frac{\bar\eta}{2} \alpha_{i+1} + \delta^*, \quad  i \in [1:D-1],\label{eq:i}\\
\eta&\leq E \frac{\bar\eta}{2} \alpha_1+ \delta^*.\label{eq:1:1}
\end{align}
\end{proposition}
This proposition follows simply by comparing consecutive columns of the matrix \eqref{eq:pyn}.
Letting $B \doteq \frac{2}{\bar\eta}$ and $C \doteq \frac{2\eta}{\bar \eta}$, the system of inequalities given in Proposition~\ref{prop:dp} can be expressed as
\begin{align}
\alpha_D &\leq B \delta^*\label{eq:general:alphaD2},\\
 \alpha_{i+1} &\geq \frac{\alpha_{i}-B\delta^*}{E}, \qquad 1 \leq i \leq D-1\label{eq:general:alphai},\\
  \alpha_{1} &\geq \frac{C-B\delta^*}{E}\label{eq:general:alpha1}.
\end{align}
We now present a series of lower bounds on $\delta^*$ that we term \emph{Type-I} lower bounds. 
\begin{lemma}\label{Lemm:lb1}
For the noise distribution given by Proposition~\ref{prop:noiseD}, we have 
\begin{align}
   \delta^* \geq \delta_k \doteq\frac{C\sum_{j=0}^{k-1}E^j-E^k}{B\sum_{j=0}^{k-1}E^{j}(j+1)}, \qquad k\in[D].\label{eq:general:type1} 
\end{align}
\end{lemma}
\begin{proof}
First note that, since $\alpha_1\leq 1$, we have from \eqref{eq:general:alpha1}, 
\begin{align}
\delta^* \geq \delta_1.
\end{align}
Since $\alpha_2 \leq 1-\alpha_1$, a second lower bound is obtained from \eqref{eq:general:alphai} for $i = 2$   by eliminating both $\alpha_2$ and $\alpha_1$ as follows
\begin{align}
\frac{\alpha_{1}-B\delta^*}{E} \leq \alpha_{2} \leq 1-\alpha_1 \quad &\Rightarrow \quad \alpha_1(1+E)-E \leq B \delta^* \\
\Rightarrow \frac{C-B\delta^*}{E} (1+E) - E \leq B \delta^* \quad &\Rightarrow \quad \delta^* \geq \delta_2.
\end{align}
Generalizing this for $1\leq k \leq D$, we obtain a general class of lower bounds as follows
\begin{align}
&\frac{\alpha_{k-1} - B\delta^*}{E} \leq \alpha_k \leq 1-\alpha_1-\cdots\alpha_{k-1}\\ &\Rightarrow \qquad \alpha_{k-1}(1+E)+E\alpha_{k-2}+\cdots+E\alpha_1 - E \leq B\delta^*\\
&\Rightarrow \alpha_{k-2}(1+E+E^2)+E^2\alpha_{k-3}+\cdots+E^2\alpha_1 -E^2\leq  B(1+2E)\delta^*.
\end{align}
Iterating this process, we obtain 
\begin{align} &\alpha_{1}\sum_{j=0}^{k-1}E^j -E^{k-1}\leq  B\delta^*\sum_{j=0}^{k-2}E^{j}(j+1)\\
&\Rightarrow \left(\frac{C-B\delta^*}{E}\right)\sum_{j=0}^{k-1}E^j -E^{k-1}\leq  B\delta^*\sum_{j=0}^{k-2}E^{j}(j+1), 
\end{align}
which upon rearranging implies $\delta^* \geq \delta_k$.
\end{proof}
In general, none of these lower bounds dominates another. We will shortly elaborate on the relationship between $\delta_k$ for $k\in [D]$. Before doing so, we first obtain another lower bound for $\delta^*$, termed a \emph{Type-II} lower bound.
\begin{lemma}\label{Lemm:lb2}
For the noise distribution given by Proposition~\ref{prop:noiseD}, we have 
\begin{align}
\delta^* &\geq \delta_{D+1} \doteq\frac{1}{B\sum_{j=0}^{D-1}E^{j}(D-j)}.\label{eq:general:type2}
\end{align}
\end{lemma}
\begin{proof}
First note that we can write from \eqref{eq:D:1} that  $\alpha_D \leq B \delta^*$.
We can extend this inequality for all $\alpha_i$ via \eqref{eq:i} as follows
\begin{equation}\label{eq:lb_alpha_i}
    \alpha_{D-k} \leq B\delta^*\sum_{j=0}^kE^j,
\end{equation}
for $k\in [0:D-1]$.
Since $\sum_{i=1}^D\alpha_i =1$, the above inequality implies that 
\begin{align*}
    1 &=  \sum_{k=0}^{D-1}\alpha_{D-k}\leq B\delta^*\sum_{k=0}^{D-1}\sum_{j=0}^{k}E^j=  B\delta^*\sum_{j=0}^{D-1}\sum_{k=j}^{D-1}E^j = B\delta^*\sum_{j=0}^{D-1}E^j(D-j),
\end{align*}
from which the result follows by a arrangement. 
\end{proof}
Putting lemmas~\ref{Lemm:lb1} and \ref{Lemm:lb2} together, the following theorem is straightforward.
\begin{theorem}\label{theoremLB}
 For the noise distribution given by Proposition~\ref{prop:noiseD}, we have
$$\delta^* \geq \tilde \delta,$$
where $\tilde \delta \triangleq \max_{k \in [D+1]}{\delta_k}$ and $\delta_k$'s were defined in \eqref{eq:general:type1} for $k\in [D]$ and in \eqref{eq:general:type2} for $k = D+1$. 
\end{theorem}

\bigskip 
To show that this lower bound is in fact tight, we need to carefully investigate the dynamic of $\delta_k$'s. To do so, we define \emph{crossover} values $C_k$ as
\begin{align}\label{eq:C:def}
C_k &\doteq \frac{\sum_{j=0}^k E^j}{\sum_{j=0}^{k-1}E^j(k-j)}, \quad k\in [D].
\end{align}
For any $k\in [D]$, the relational order between two consecutive lower bounds $\delta_k$ and $\delta_{k+1}$ is uniquely determined by the relation of the parameter $C\doteq \frac{2\eta}{\bar \eta} $ with the crossover value $C_k$ as prescribed in the following lemma. This, together with strictly decreasing behavior of the sequence $\{C_k\}_{k=1}^{D}$ (to be proved in Lemma \ref{lem:monotonic}) will ensure that we can find the maximum among all $\delta_k$, $k \in [D+1]$ analytically. 
\begin{lemma}\label{lem:delta:relations}
We have $\delta_k \geq \delta_{k+1}$ for $k\in [D]$ if and only if $\eta$ is such that $C \doteq \frac{2\eta}{\bar \eta} \geq C_k$.
\end{lemma}
\begin{proof}
First consider $k = 1$ and note that 
\begin{align}
\delta_1 - \delta_{2} \geq 0 \quad\Leftrightarrow\quad C-E \geq \frac{C+CE-E^2}{1+2E} \quad\Leftrightarrow\quad C\geq 1+E = C_1.
\end{align}
For $k \in [2:D-1]$, it follows that $\delta_k \geq \delta_{k+1}$ if and only if 
\begin{align*}
\frac{C\sum_{j=0}^{k-1}E^j-E^k}{B\sum_{j=0}^{k-1}E^{j}(j+1)} &\geq \frac{C\sum_{j=0}^{k}E^j-E^{k+1}}{B\sum_{j=0}^{k}E^{j}(j+1)}.
\end{align*}
After a straightforward algebraic manipulation, we can show that it is equivalent to  
$$C \geq \frac{\sum_{j=0}^k E^j}{\sum_{j=0}^{k-1}E^j(k-j)} = C_k.$$
Last, we have $\delta_D \geq  \delta_{D+1}$ if and only if
\begin{align*}
&\frac{C\sum_{j=0}^{D-1}E^j-E^D}{B\sum_{j=0}^{D-1}E^{j}(j+1)} \geq \frac{1}{B\sum_{j=0}^{D-1}E^{j}(D-j)},
\end{align*}
which is in turn equivalent to 
\begin{align*}
&C\left(\sum_{j=0}^{D-1}E^j\right)\left(\sum_{j=0}^{D-1}E^{j}(D-j)\right)
\geq E^D\left(\sum_{j=0}^{D-1}E^{j}(D-j)\right)+\sum_{j=0}^{D-1}E^{j}(j+1) \\
&= \left(\sum_{j=0}^{D-1}E^j\right)\left(\sum_{j=0}^{D}E^j\right).
\end{align*}
Hence, we conclude $\delta_D \geq  \delta_{D+1}$ if and only if 
$$ C \geq \frac{\sum_{j=0}^{D-1}E^j}{\sum_{j=0}^{D-1}E^{j}(D-j)} = C_D.$$
\end{proof}
 \begin{lemma}\label{lem:monotonic}
The sequence $\{C_k\}_{k=1}^{D}$ is strictly decreasing.
\end{lemma}
\begin{proof}
This can be easily verified by writing the difference between $C_k$ and $C_{k+1}$ and checking that
\begin{align}
C_k - C_{k+1} = \frac{\sum_{j=0}^{k}E^j(j+1)}{(\sum_{j=0}^{k-1}E^j(k-j))(\sum_{j=0}^{k}E^j(k+1-j))} > 0, \quad k\in [D].
\end{align}
\end{proof}
For notational simplicity, define $C_{D+1} \doteq 0$ and $C_0 \doteq \infty$. Thus, $\{C_k\}_{k=0}^{D+1}$ is still monotonically decreasing. Combining Lemma~\ref{lem:delta:relations} with Lemma~\ref{lem:monotonic}, we obtain the following result.
\begin{theorem}\label{Thm:cor}
For the noise distribution given by Proposition~\ref{prop:noiseD}, we have
$$\tilde \delta \doteq \max_{i \in [D+1]}{\delta_i} = \delta_k,$$ 
if and only if $\eta$ is such that $ C_k < C \leq C_{k-1}$.
\end{theorem}
\begin{proof}
To prove the only-if direction, note that due to strict monotonicity of $\{C_k\}_{k=0}^{D+1}$, $C > C_k$ implies $C > C_{k+1} > \cdots> C_{D}$. Therefore, according to Lemma~\ref{lem:delta:relations}, $C > C_k$ implies $\delta_k> \delta_{k+1} > \delta_{k+2}> \dots> \delta_{D+1}$. Similarly, $C \leq C_{k-1}$  implies $C < C_{k-2} < \cdots< C_{1}$. Therefore, $C \leq C_{k-1}$  implies $\delta_k \geq \delta_{k-1} > \delta_{k-2}> \dots> \delta_{1}$. In summary, $ C_k < C \leq C_{k-1}$ implies $\tilde \delta =\delta_k$. The forward direction can be similarly argued.
\end{proof}
Together with Theorem~\ref{theoremLB},  this theorem implies that if $C_k<C\leq C_{k-1}$, then   
$$\delta^*\geq \delta_k.$$
Next, we design a mechanism for which this lower bound is achieved. 
\subsection{Achievability of Lower Bounds}
In this section, we show that the lower bound $\tilde \delta = \max_{i \in [D+1]}{\delta_i}$ on $\delta^*$ is in fact achievable, i.e., there exists $\{\alpha_i\}_{i=1}^D$ with $\alpha_i\geq 0$ and $\sum_{i=1}^D\alpha_i=1$ such that the mechanism \eqref{eq:pyn} satisfies the constraints of the linear program \eqref{Eq:LP_DP} with $\delta = \tilde\delta$.
\begin{definition}[Optimal $\alpha$'s]\label{def:alpha}
Set $\tilde \delta = \max_{i \in [D+1]}{\delta_i}$ as given in Theorem \ref{theoremLB}. If $\eta$ is such that $0 < C \leq C_D $, then $\tilde \delta = \delta_{D+1}$ and we define $\alpha_j^*$'s as follows
\begin{align}\label{eq:opt:alpha:type2}
\alpha_D^* &= B\tilde\delta = \frac{1}{\sum_{\ell = 0}^{D-1}E^\ell (D-\ell)},\\
\alpha_{j}^* &= E\alpha_{j+1}^* + B\tilde\delta = \frac{\sum_{\ell=0}^{D-j}E^\ell}{\sum_{\ell = 0}^{D-1}E^\ell (D-\ell)},\qquad j\in [1:D-1].\label{eq:opt:alpha:type22}
\end{align}
If $\eta$ is such that $C > C_D$, then  $\tilde \delta = \delta_{k}$ for some $k \in [D]$. We then define $\alpha_j^*$'s for $j \in[k]$ as
\begin{align}\label{eq:opt:alpha:type1}
\alpha_1^* = \frac{C-B\tilde\delta}{E}, \qquad \alpha^*_j = \frac{\alpha_{j-1}^*-B\tilde\delta}{E}, \qquad j \in [2:k]
\end{align}
and set $\alpha^*_j = 0$ for $j \in [k+1:D]$.\footnote{In the proof of Theorem \ref{thm:optimal}, we will show $\alpha_1^*> 0$, $\alpha_j^* \geq 0$ for $j \in [2:k]$ in \eqref{eq:opt:alpha:type1} and that $\sum_{j=1}^D \alpha_j^* = 1$ as desired.}
\end{definition}
Now we are ready for the main result of this section.
\begin{theorem}\label{thm:optimal}
The mechanism \eqref{eq:pyn} with coefficients $\{\alpha_i^*\}$ defined in Definition \ref{def:alpha} satisfies the constraints of the linear program \eqref{Eq:LP_DP} with  $\delta = \tilde\delta$, defined in Theorem \ref{theoremLB}. In particular, 
$$\delta ^* = \tilde\delta.$$
\end{theorem}

\begin{proof}
First note that once we prove that $\tilde\delta$ is achievable, Theorem~\ref{theoremLB} implies that $\delta^* = \tilde\delta$. Thus, we only need to prove the achievability of $\tilde\delta$. Note further that $\{\alpha^*_j\}$ satisfy the inequalities in Proposition~\ref{prop:dp}, with $\delta^*$ replaced by $\tilde\delta$, with equality. Therefore, we only need to show that $\{\alpha^*_i\}$ are valid in the sense that they are non-negative and sum to 1.

 Assume $C > C_D$ in general and in particular, we have $C_k < C \leq C_{k-1}$, for some $k \in [D]$ (recall that $C_0 = \infty$). It follows from Theorem~\ref{Thm:cor} that $\tilde\delta = \delta_k$, where $\delta_k$ is a Type-I lower bound given in \eqref{eq:general:type1}. 
First, we note that according to \eqref{eq:opt:alpha:type1}
\begin{align}
\alpha_1^* = \frac{C-B\tilde\delta}{E} = \frac{C\sum_{\ell=0}^{k-2}E^\ell(\ell+1)+E^{k-1}}{\sum_{\ell=0}^{k-1}E^\ell(\ell+1)} >0.
\end{align}
Now consider $\alpha_k^*$ in \eqref{eq:opt:alpha:type1}. Applying \eqref{eq:opt:alpha:type1} iteratively, we obtain 
\begin{align}\label{eq:alpha:k:positive}
    \alpha_k^* &= \frac{\alpha_{k-1}^*-B\tilde\delta}{E} = \frac{C-B\delta_k\sum_{\ell=0}^{k-1}E^\ell}{E^k}\\&= \frac{\sum_{\ell=0}^{k-1}E^\ell-C\sum_{\ell=0}^{k-2}E^\ell(k-1-\ell)}{{\sum_{\ell=0}^{k-1}E^\ell(\ell+1)}}\\
    &\geq \frac{\sum_{\ell=0}^{k-1}E^\ell-C_{k-1}\sum_{\ell=0}^{k-2}E^\ell(k-1-\ell)}{{\sum_{\ell=0}^{k-1}E^\ell(\ell+1)}} = 0,\label{eq:alpha:k:positive2}
\end{align}
where in the last two steps we have used the fact that $C \leq C_{k-1}$ and the definition of $C_{k-1}$ in \eqref{eq:C:def}. 
Note that the definition of $\{\alpha_i^*\}$ in  \eqref{eq:opt:alpha:type1} together with\footnote{Note that from \eqref{eq:general:type2} we conclude $\delta_{D+1} > 0$. Therefore, $\tilde\delta = \max_i \delta_i > 0$.} $\tilde \delta > 0$ implies that $\alpha^*_{i} < \alpha^*_{j}$ for $2\leq j<i\leq D$. In particular, we have $\alpha_{k-j}^* > \alpha_{k}^* \geq 0$ for $j \in [k-1]$. 
Note that 
 \begin{align}
\sum_{i=1}^{D} \alpha_i^* = \sum_{i=1}^{k} \alpha_i^* = \frac{C\sum_{j=0}^{k-1}E^j - B\delta_k\sum_{j=1}^kjE^{j-1}}{E^k} = 1,
\end{align}
where we use the fact that $\alpha_j^* = 0$ for $j \in [k+1:D]$.

Finally, if $C \leq C_D$ then from Theorem~\ref{Thm:cor}, we conclude $\delta^* = \delta_{D+1}$, which is the Type-II lower bound on $\delta^*$. In this case, $\{\alpha^*_i\}$ are given in \eqref{eq:opt:alpha:type2} and \eqref{eq:opt:alpha:type22}. A similar argument as above shows that 
$0 \leq \alpha^*_i \leq 1$ for each $i\in [D]$ and also $\sum_i \alpha^*_i = 1$.

\end{proof}

This theorem demonstrates that $\{\alpha^*_i\}$, defined in Definition~\ref{def:alpha}, constructs the ``optimal" mechanism when plugging into \eqref{eq:pyn} or equivalently the optimal noise distribution when plugging into \eqref{noise:ngeqD}. 
Fig.~\ref{noisepmf} shows such optimal noise distribution for $\eps = 2.18$ ($E = 8.8463$), $\eta = 0.8$ ($C = 8$), and $D = 6$ for $z \geq 0$ (noting the symmetry for $z < 0$). Note that the minimum true count is assumed to be $n \geq D = 6$, which seems a reasonable assumption in large datasets with hundreds of participants. Using \eqref{eq:C:def}, we find that $7.8867 = C_3 < C \leq C_2 = 8.1229$. Therefore, $\tilde \delta = \delta_3 = 0.0049$ according to Theorem \ref{Thm:cor}. It follows from \eqref{eq:opt:alpha:type1} that only $\alpha_1^*$, $\alpha_2^*$, and $\alpha_3^*$ are non-zero. 
Consequently, the noise distribution is in fact supported on $[-3: 3]$, instead of   what we originally required in P1, i.e., $[-6: 6]$.

Using Equation (1) in \cite{count:dp:gaussian:arxiv}, we also plot in Fig.~\ref{noisepmf}, discrete Gaussian distribution, denoted by $p_{\text{G},Z}(z)$, for $z \geq 0$. We set its variance identical to that of our noise distribution, i.e., $\sigma^2 = \bar\eta \sum_{i=1}^D \alpha_i^* i^2$. It is worth noting that discrete Gaussian distribution is not compactly supported; however, the probability that it assigns to points outside $[-D:D]$ is negligible and is not shown in Fig.~\ref{noisepmf}. While the two distributions may look somewhat similar, they have important differences that substantially impact their privacy performance. Notice, for instance, that 
$p_{\text{G},Z}(\pm 1) = 0.11685$ and $p_{\text{G},Z}(\pm 2) = 0.000416$, and hence $\frac{p_{Y}(n+2|n+1)}{p_{Y}(n+2|n)} = \frac{p_{Y}(n-1|n)}{p_{Y}(n-1|n+1)}$ can be as large as about 282. Note that in our proposed distribution, $p_{Z}(\pm 1) = \frac{\bar\eta}{2}\alpha_1^* =  0.08987$ and $p_{Z}(\pm 2) = \frac{\bar\eta}{2}\alpha_2^* =  0.00960$, and hence $\frac{p_{Y}(n+2|n+1)}{p_{Y}(n+2|n)} =\frac{p_{Y}(n-1|n)}{p_{Y}(n-1|n+1)} = \frac{\alpha^*_1}{\alpha^*_2} = \frac{0.08987\cdots}{0.00960\cdots} = 9.3617$. 
Consequently, in the discrete Gaussian mechanism the smallest $\eps$ such that $p_{Y}(n-1|n)\leq e^\eps p_{Y}(n-1|n+1) + \tilde \delta$ for $\tilde\delta = 0.0049$  is $\eps = 5.6$, while for our mechanism $\eps = 2.18$ is sufficient. This indicates that the optimally-chosen coefficient $\{\alpha_i\}$ in our setting improves the privacy guarantee of discrete Gaussian mechanism. In the next section, we provide more detailed comparison between there two mechanisms.  

\begin{figure}[h]\begin{center}
\includegraphics[width=0.6\columnwidth]{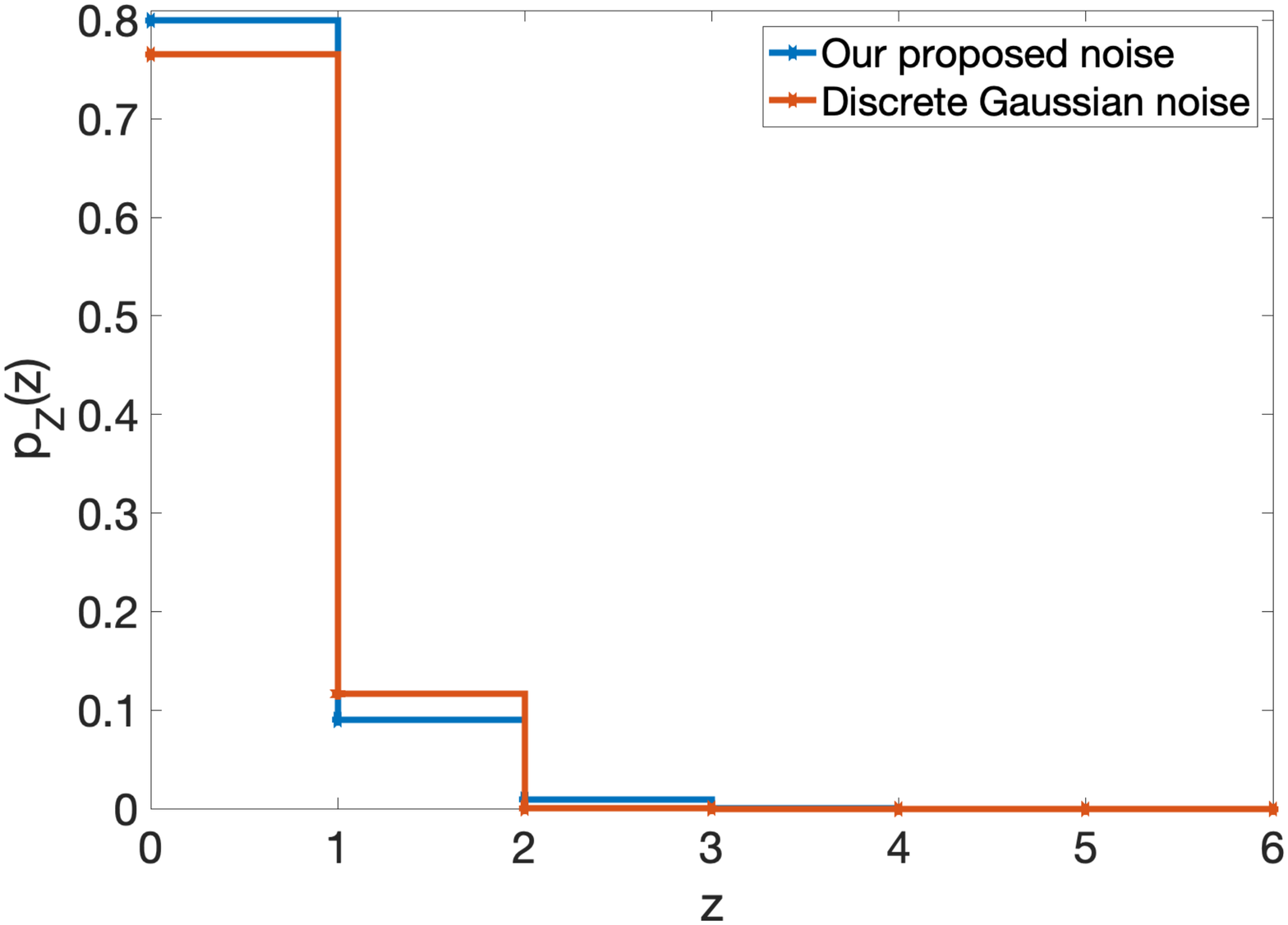}\caption{The noise distribution \eqref{noise:ngeqD} with $\{\alpha_i\}$ being chosen according to Definition \ref{def:alpha} is compared with the discrete Gaussian noise \cite{count:dp:gaussian:arxiv} with the identical variance. Here, we assume 
for $\eps = 2.18$, $\eta = 0.8$, and $D = 6$.  The more `graceful` transition in our noise distribution from $z=1$ to $z=2$ results in an improved privacy performance. }\label{noisepmf}
\end{center}
\end{figure}

\begin{remark}\label{remark:Singular_to_DP}
As mentioned in Remark~\ref{remark:DP_Singular}, we restricted the definition of differential privacy in Definition~\ref{Def:DP} to the singular events. As a result, the parameters $\eps$ and $\delta$ computed in this section do not correspond to the DP parameters. To compute the relationship between these parameters, recall that the mechanism \eqref{eq:pyn} is $(\eps, \delta)$-DP if for any event $S\subset [N+D]$
$$p_{Y|n}(S)\leq e^\eps p_{Y|n+1}(S) + \delta \qquad\text{and}\qquad  p_{Y|n}(S)\leq e^\eps p_{Y|n-1}(S) + \delta,$$
for all possible query responses $n\in [N]$. Since $p_{Y|n}(S) = \sum_{y\in S}p_{Y|n}(y)$, our analysis in this section implies that for any $S\subset[N+D]$
\begin{align*}
    p_{Y|n}(S) - e^\eps p_{Y|n+1}(S) &= \sum_{y\in S}\left[p_{Y|n}(y)- e^\eps  p_{Y|n+1}(y)\right] \\
    &\leq \sum_{y\in S\cap\text{supp}(p_{Y|n})}\left[p_{Y|n}(y)- e^\eps  p_{Y|n+1}(y)\right]\\
    &\leq |S\cap\text{supp}(p_{Y|n})| \max_{y\in [N+D]}\left[p_{Y|n}(y)- e^\eps  p_{Y|n+1}(y)\right]\\
    &\leq (2D+1) \max_{y\in [N+D]}\left[p_{Y|n}(y)- e^\eps  p_{Y|n+1}(y)\right]\\
    &\leq (2D+1)\delta^*,
\end{align*}
where $\text{supp}(p_{Y|n})$ denotes the set of $y\in [N+D]$ with $p_{Y|n}(y)>0$ and a closed-form expression for $\delta^*$ was given in Theorem~\ref{thm:optimal}. The same argument shows that $p_{Y|n}(S) - e^\eps p_{Y|n-1}(S) \leq (2D+1)\delta^*$. This observation indicates that the mechanism  \eqref{eq:pyn} is $(\eps, \min\{1,(2D+1)\delta^*\})$-DP when $\{\alpha_i\}$ is chosen according to Definition~\ref{def:alpha}. 
\end{remark}
\section{Numerical Results}\label{sec:numerical}
We begin this section by numerically computing the DP parameters $\eps$ and $(2D+1)\delta^*$ of our mechanism according to  Theorem~\ref{thm:optimal}.\footnote{We mostly consider practical parameter ranges resulting in $(2D+1)\delta^* \ll 1$. We thus write $(2D+1)\delta^*$ instead of $\min\{1,(2D+1)\delta^*\}$.} In Fig.~\ref{test1}, we assume $\eta = 0.5 $ and plot the $(2D+1)\delta^*$ in terms of $\eps$ for different values of $D$. 
It is clear that $(2D+1)\delta^*$ is no greater than $0.001$ for $D=8$ and $\eps > 1.1$. 
It is worth noting that while we reveal the true count with probability $0.5$, from Definition \ref{def:alpha}, the mechanism's output lies in an acceptable range of $[n-3:n+3]$ with a much higher probability $\eta + \bar\eta(\alpha_1^*+\alpha_2^*+\alpha_3^*) = 0.9945$ for $\eps = 1.5$.
\begin{figure}[h]
\centering
\includegraphics[width=0.6\columnwidth]{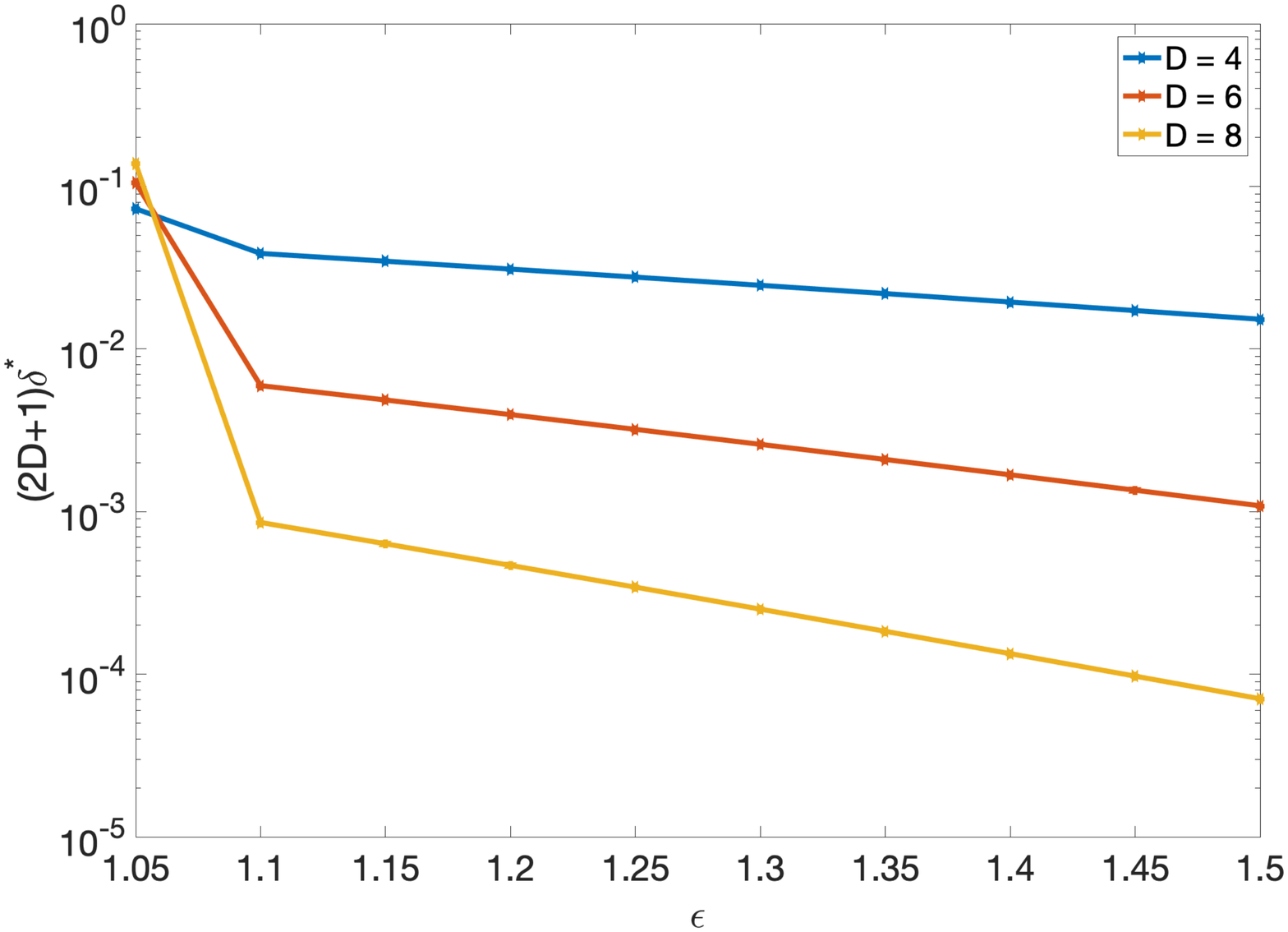}\caption{DP parameters $(2D+1)\delta^*$ versus $\eps$ for $\eta = 0.5 $ and different values of $D$. }\label{test1}
\end{figure}

In Fig.~\ref{test3}, we illustrate the relationship between the DP parameter $(2D+1)\delta^*$ and $\eta$ for different values of $\eps$ and $D$. 
Quite predictably, for given $D$ and $\eps$, there exists a threshold for $\eta$ above which $(2D+1)\delta^*$ grows fast to one; thus indicating a trade-off between utility (reliability) and privacy. To strike a good balance between the reliability and privacy, one may choose the the largest value of $\eta$ just before the `knee' phenomenon occurs. For instance, assuming $D = 8$, our mechanism provides $(1.1, 10^{-3})$-DP and $(2.2, 5\times10^{-7})$-DP guarantee while presenting $\eta = 0.5$ and $\eta = 0.8$, respectively. We remark that $\eps = 1.1$ and $\eps = 2.2$ were numerically chosen such that the knee phenomenon occurs slightly after $\eta = 0.5$ and $\eta = 0.8$, respectively.

\begin{figure}[h]
\centering
\includegraphics[width=0.45\columnwidth,height=0.4\columnwidth]{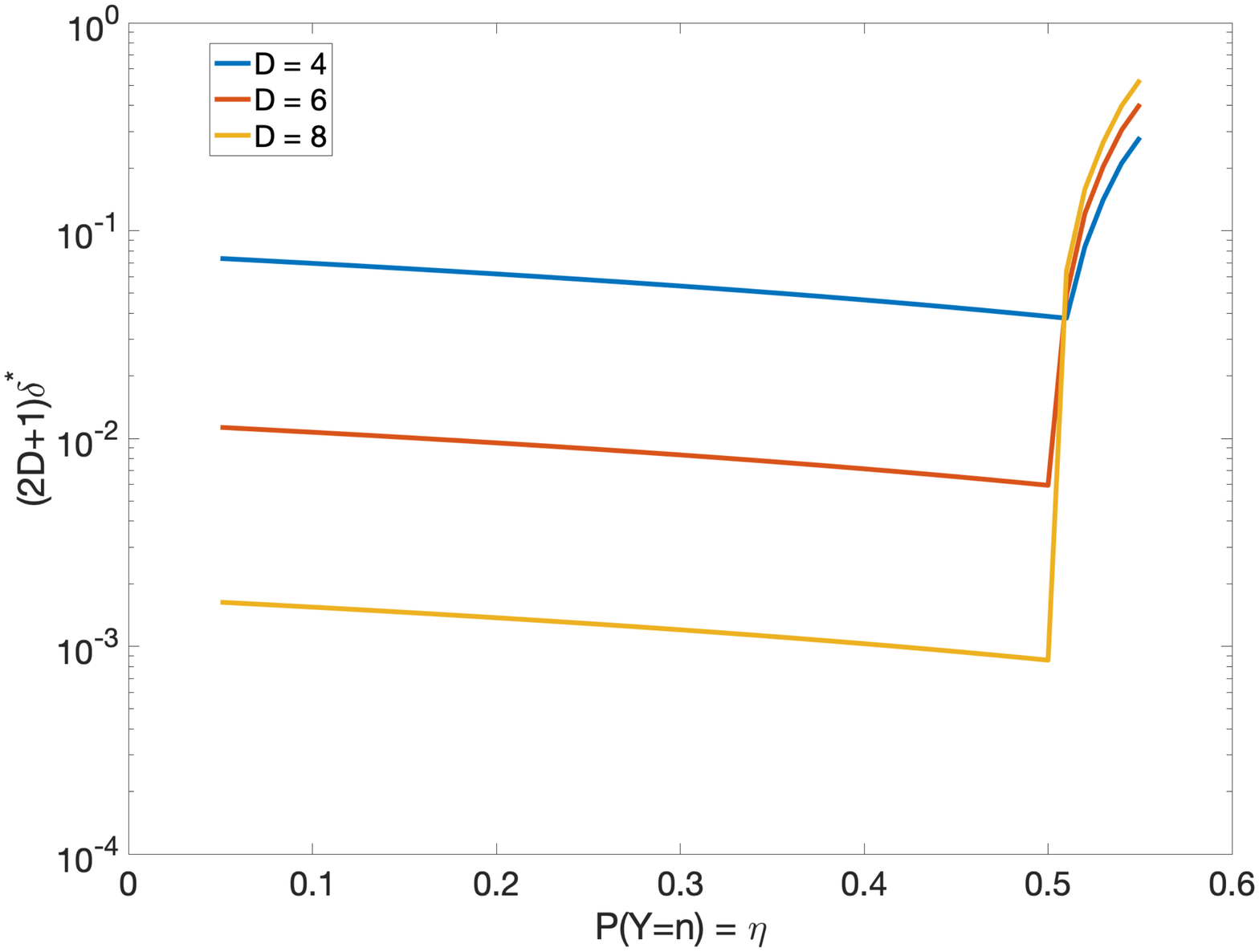}
\includegraphics[width=0.45\columnwidth, height=0.4\columnwidth]{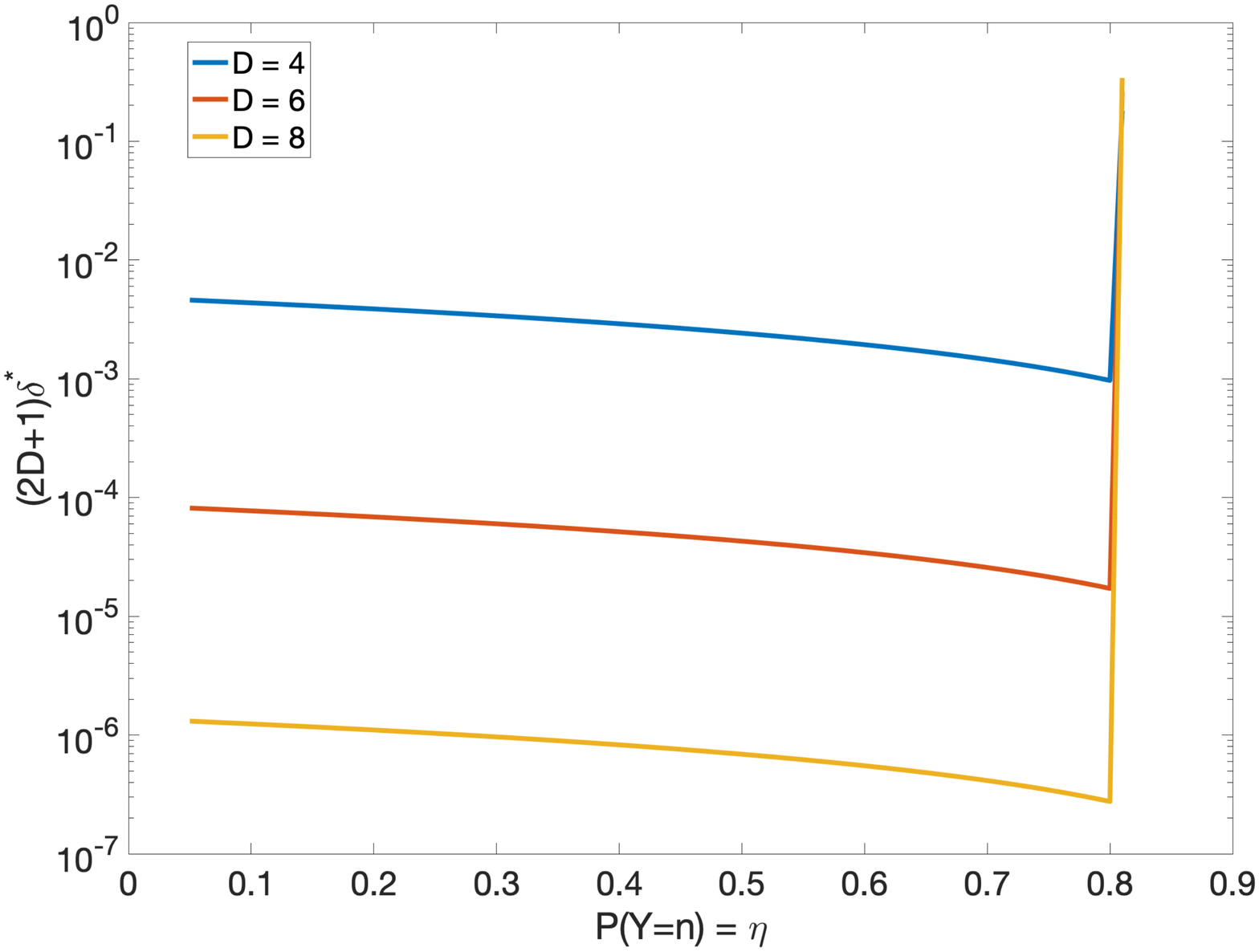}
\caption{DP parameter $(2D+1)\delta^*$ versus $\eta$ for $\eps = 1.1 $ (left) and $\eps = 2.2$ (right) and different values of $D$. }
\label{test3}
\end{figure}
Next, we compare our mechanism with the recently proposed \textit{discrete} Gaussian mechanism \cite{count:dp:gaussian:arxiv} in terms of the utility-privacy performance.  
There may be different ways for conducting such comparison. To have a fair comparison, we take the following steps:  
\begin{enumerate}
    \item Given fixed values of $\eta \in (0,1)$, $D \in \mathbb{N}$, and $\eps\geq 0$, we compute the optimal coefficient $\{\alpha^*_i\}$ (according to Definition \ref{def:alpha}) and the corresponding DP parameter $(2D+1)\delta^*$  (according to Theorem \ref{thm:optimal} and Remark \ref{remark:Singular_to_DP}). We then compute $\sigma^2(\eta, D, \eps)$ the noise variance given by  $$\sigma^2(\eta, D, \eps) \doteq \bar\eta \sum_{i=1}^D \alpha_i^* i^2.$$
    \item We generate a discrete Gaussian probability mass function (pmf) according to Equation (1) in \cite{count:dp:gaussian:arxiv} with the variance $\sigma^2(\eta, D, \eps)$. This dictates that the discrete Gaussian mechanism has the same utility as our mechanism. We refer to the resulting distribution as $p_{\text{G},Z}$.  
    \item Let $\eps_\text{G}$ and $\delta_\text{G}$ be the DP parameters of the discrete Gaussian mechanism. It is proved in \cite[Theorem 7]{count:dp:gaussian:arxiv} that
    \begin{align}\label{delta_G}
        \delta_\text{G} = p_{\text{G}, Z}\left(Z > \eps_\text{G}\sigma^2-0.5\right)-e^{\eps_\text{G}}p_{\text{G}, Z}\left(Z > \eps_\text{G}\sigma^2+0.5\right),
    \end{align}
    where $\sigma^2$ is the variance of the discrete Gaussian noise added to the true count. By replacing $\sigma^2$ by $\sigma^2(\eta, D, \eps)$, we make use of this expression to compare the discrete Gaussian mechanism with our mechanism in terms of privacy parameters.  
\end{enumerate}

\begin{figure}[h]
\begin{center}
\includegraphics[width=0.6\columnwidth]{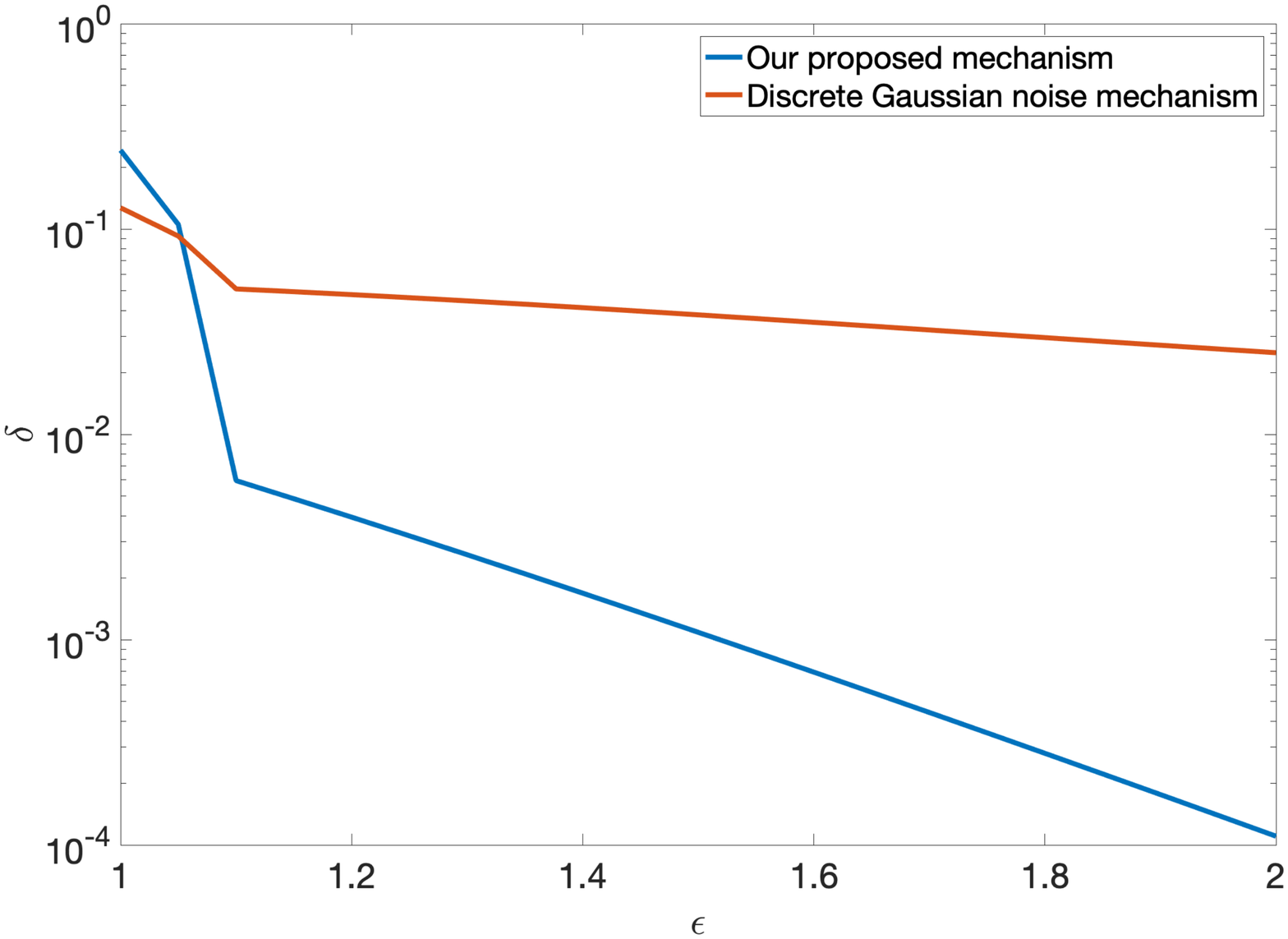}\caption{Comparison of $(\eps, \delta)$-DP performance for $\eta = 0.5$ and $D = 6$ in our proposed mechanism with that of a discrete Gaussian mechanism with the same variance $\sigma^2(\eta, D, \eps)$. In our mechanism, for a given $\eps\geq 0$, we compute the optimal coefficient $\{\alpha^*_i\}$ (according to Definition \ref{def:alpha}) and the corresponding DP parameter $(2D+1)\delta^*$  (according to Theorem \ref{thm:optimal} and Remark \ref{remark:Singular_to_DP}). For the discrete Gaussian mechanism, we use \eqref{delta_G} to compute $\delta_G$.}\label{test4}
\end{center}
\end{figure}

As depicted in Fig \ref{test4}, the $(\eps, \delta)$-DP performance of our proposed scheme is superior to that of the discrete Gaussian mechanism with the same variance for a wide range of privacy parameters (both mechanisms perform rather poorly for $\eps < 1.1$, resulting in quite high $\delta$). Many other experiments were conducted and they all showed  trends similar to what is shown in Fig. \ref{test4}. 
\section{Conclusion}
In this work, we explicitly constructed a privacy-preserving mechanism for responding to integer-valued queries to a dataset containing sensitive attributes.  
This mechanism is noise-additive but, unlike many other noise-additive private mechanisms, it adds an integer-valued noise to the true count. The  noise distribution is parameterized by $\eta\in (0,1)$ which specifies the allowable probability of error in responding the true count, thereby  balancing the privacy-utility trade-off.  
We proved that this mechanism is $(\eps, \delta)$-differentially private where both $\eps$ and $\delta$ depend on noise support size and $\eta$. This mechanism is contrasted with the recently proposed discrete Gaussian mechanism \cite{count:dp:gaussian:arxiv} which adds discretized Gaussian noise (with infinite support) to the true count.  
Our numerical findings indicate that the resulting $\eps$ and $\delta$ are tighter than those obtained in discrete Gaussian mechanism for many practical range of $\eta$.

We conclude this work with a note on future direction. In this work, we assume that there is a single query to a dataset. However, in many practical scenarios there are several queries to a single dataset, each of which might depend on the previous ones. Provided that each query is responded by a our private mechanism, it is essential to determine how privacy degrades as the number of queries increases. We believe that the new technique propose in \cite{hundredrounds} might be helpful for this purpose.

\bibliographystyle{IEEEtran}
	\bibliography{References}
\end{document}